\documentclass[conference]{IEEEtran}
\usepackage{subcaption}

\usepackage[utf8]{inputenc}
\usepackage{tikz}
\usepackage{amsmath,amssymb,amsthm}

\usepackage{csvsimple}
\usepackage{cite}

\usepackage{tikz}
\usetikzlibrary{arrows}
\usetikzlibrary{positioning}  
\usetikzlibrary{shadows}
\usetikzlibrary{arrows,backgrounds,calc,chains,
matrix,positioning,shapes,shapes.geometric,
shapes.arrows, decorations.pathmorphing, decorations.pathreplacing}

\newcommand{\B}{\mathbb{B}}

\newtheorem{definition}{Definition}
\newtheorem{example}{Example}
\newtheorem{theorem}{Theorem}
\newtheorem{corollary}{Corollary}

\usepackage{etoolbox}
\makeatletter
\patchcmd{\maketitle}{\@copyrightspace}{}{}{}
\makeatother

\usetikzlibrary{backgrounds,calc,chains,matrix,positioning,shapes,shapes.geometric,shapes.arrows}

\tikzset{
	font={\footnotesize},
	vertex/.style={draw,circle,inner sep=0pt,minimum width=0.5cm,minimum height=0.5cm},
	zeroterm/.style={below,inner sep=0pt,font=\tiny}
}

\title{\huge{One Additional Qubit is Enough:} \\ \LARGE{Encoded Embeddings for Boolean Components in Quantum Circuits\vspace{-.4cm}}\thanks{This work has partially been supported by the European Union through the COST Action IC1405 and the Google Research Award Program.}}

\newcommand{\ket}[1]{\ensuremath{\left|#1\right\rangle}}

\usepackage{url}
\usepackage{qcircuit}

\renewcommand{\baselinestretch}{.96}

\usepackage[keeplastbox]{flushend}

\author{\IEEEauthorblockN{Alwin Zulehner\IEEEauthorrefmark{1}, Philipp Niemann\IEEEauthorrefmark{2}, Rolf Drechsler\IEEEauthorrefmark{2}\IEEEauthorrefmark{3}, and Robert Wille\IEEEauthorrefmark{1}}
	\IEEEauthorblockA{\small \IEEEauthorrefmark{1}Institute for Integrated Circuits, Johannes Kepler University Linz, Austria}
	\IEEEauthorblockA{\small \IEEEauthorrefmark{2}Cyber-Physical Systems, DFKI GmbH, Bremen, Germany\\
		}
	\IEEEauthorblockA{\small \IEEEauthorrefmark{3}Department of Computer Science, University of Bremen, Bremen, Germany\\
		}
	\IEEEauthorblockA{\small alwin.zulehner@jku.at, philipp.niemann@dfki.de, drechsle@informatik.uni-bremen.de, robert.wille@jku.at}\vspace{-.6cm}
}

\begin{document}
\maketitle


\begin{abstract}

Research on quantum computing has recently gained significant momentum since first physical devices became available. 
Many quantum algorithms make use of \mbox{so-called} \emph{oracles} that implement Boolean functions and are queried with highly superposed input states in order to evaluate the implemented Boolean function for many different input patterns in parallel.
To simplify or enable a realization of these oracles in quantum logic in the first place, the Boolean reversible functions to be realized usually need to be broken down into several non-reversible sub-functions. However, since quantum logic is inherently reversible, these sub-functions have to be realized in a reversible fashion by adding further qubits in order to make the output patterns distinguishable (a process that is also known as \emph{embedding}). This usually results in a significant increase of the qubits required in total.
In this work, we show how this overhead can be significantly reduced by utilizing coding. More precisely, we prove that one additional qubit is always  enough to embed \emph{any} non-reversible function into a reversible one by using a variable-length encoding of the output patterns. Moreover, we characterize those functions that do not require an additional qubit at all.
The made observations show that coding often allows one to undercut the usually considered minimum of additional qubits in sub-functions of oracles by far.
\end{abstract}

\section{Introduction}
\label{sec:intro}

Quantum algorithms running on quantum computers allow for significant (exponential in the best case) speed-ups compared to their classical counterparts by exploiting quantum-mechanical phenomena like superposition, entanglement, and phase shifts~\cite{NC:2000}. 
Recently, devices that have been made publicly available---together with the commitment of companies like IBM, Google, Microsoft, and Rigetti---brought new momentum into a domain that has been considered as a ``dream of the future'' for a long time~\cite{courtland2017google,gomes2018quantumcomputing,hsu2018quantumcomputing}. Even though these first devices are limited in qubit fidelity and their number of qubits (i.e., they are classified as NISQ devices~\cite{preskill2018quantum}), they provide a first step towards building a fault-tolerant quantum computer that is capable of conducting hard and useful tasks in \mbox{non-exponential} time. 

Many proposed quantum algorithms contain large Boolean parts (also called \emph{oracles}) that are queried with a highly superposed input to gain quantum speed-up. Examples are the modular exponentiation in Shor's algorithm for integer factorization~\cite{Sho:94} or a Boolean description of the database that is queried in Grover's Algorithm~\cite{Gro:96}. In order to use these Boolean components on a quantum computer, they have to be described as quantum circuits (i.e.,~a sequence of quantum operations that are applied to the qubits)---an inherently reversible description means. Since it is very complex to determine a sequence of quantum operations (also denoted quantum gates) that realize the desired functionality (a process termed \emph{synthesis}~\cite{DBLP:conf/ismvl/DrechslerW11,DBLP:journals/csur/SaeediM13}), the Boolean function to be realized is usually decomposed into several (not necessarily reversible) sub-functions~\cite{DBLP:journals/qic/HanerRS17,DBLP:journals/qic/Beauregard03,haener2018quantum}. Hence, even though the overall functionality of the oracle is inherently reversible, its sub-components may not be. 

In order to realize non-reversible functions in quantum logic, further qubits (often called \emph{ancillary}, \emph{ancillae}, or \emph{working} qubits) have to be added in order to make the output patterns distinguishable and, hence, obtain a reversible function (a process called \emph{embedding}~\cite{Soeken:2015:ELB:2856147.2786982,ZulehnerW17Emb}). Moreover, such additional qubits are often used to store intermediate results and have to be restored to their initial state (by de-computing intermediate results) before ``leaving'' the oracle. All this obviously increases the number of qubits needed to realize the oracle. In fact, even if the embedding process guarantees a minimum of ancillary/ancillae/working qubits, their number is frequently quite substantial---a severe drawback since qubits are a highly limited resource.

In order to overcome the issue outlined above, we propose to utilize \emph{coded} embeddings where each occurring output pattern is encoded with another (smaller) unique pattern. This way, we utilize recently proposed embedding and synthesis schemes such as \emph{one-pass synthesis} of reversible logic~\cite{zulehner2017one} as well as 
synthesis exploiting coding techniques~\cite{ZulehnerW18Exploiting,zulehner2018pushing} for the realization of quantum oracles. Encoding outputs allows us to significantly reduce the number of qubits even below what is usually considered to be the minimum. Although this changes the intended functionality, using encoded values is still acceptable for the realization of oracles since subsequent sub-components just have to be slightly adjusted to handle the code, or need to be equipped with a small decoder beforehand (which often is easier to realize than the original functionality). 

Moreover, in this work we show for the first time that utilizing all that potential indeed allows for the realization of Boolean non-reversible sub-components with at most one additional qubit only. In addition to that, we exactly identify the cases where even this additional qubit is not necessary.
By this, we can provably show that, using coding, one additional qubit is enough, i.e.,~that the proposed scheme 
often allows one to undercut the usually considered minimum of additional qubits in oracles by far. This is additionally confirmed by experimental evaluations.
Note that while we only cover the two-valued case here---since this is the \mbox{de facto} standard in quantum computation and a large set of benchmarks is available for evaluation---we expect that our (theoretical) results can be extended to the multiple-valued case with a radix $r>2$ in a straightforward fashion (e.g., using the generalization as proposed in~\cite{zulehner2018generalizing}).

\vspace{200cm}

The remainder of this work is structured as follows. In Section~\ref{sec:background}, we briefly introduce the basics of quantum circuits as well as how non-reversible functions can be realized by them. Section~\ref{sec:one_is_enough} provides a technique for encoding the function to be realized. Here, we also formally prove that using a \mbox{variable-length} code indeed allows for realizations with at most one additional qubit.
Section~\ref{sec:results} compares the number of required qubits in coded embeddings to those embeddings (without encoding) that have been considered to be the minimum thus far. Section~\ref{sec:conclusions} concludes the paper.

\section{Background}
\label{sec:background}

In this section, we briefly recap the basics of quantum circuits, as well as how to realize Boolean components occurring in them.

\subsection{Quantum Circuits}
\label{sec:circuits}

Quantum computations are conducted by applying operations to qubits---entities that cannot only be in one of its two basis states (denoted~$\ket{0}$ and $\ket{1}$), but also in an (almost) arbitrary superposition of both.
Typical operations acting on a single qubits are negating the state of a qubit (NOT operation, denoted by $X$ or $\oplus$), setting a qubit into superposition (Hadamard operation, denoted by $H$), or conducting a phase shift by $i$ (denoted by $S$). Moreover, these operations may be controlled by other qubits. Then, the operation is only conducted if all controlling qubits are in basis state $\ket{1}$.
All these computations may be represented by means of circuit diagrams, where each qubit is represented by a horizontal line and quantum gates (i.e.,~operations that are applied to the qubits) on these lines determine (from left to right) in which order the respective operations are applied to the qubits.

\begin{figure}
	\begin{subfigure}[b]{0.48\linewidth}
		\begin{center}
		\qquad
		\qquad
		\mbox{	
			\Qcircuit @C=1em @R=1em {
				\lstick{\ket{q_0} = \ket{0}} & \gate{H} & \ctrl{1} & \qw \\
				\lstick{\ket{q_1} = \ket{1}} & \qw & \targ & \qw \\
		}}
		
	\end{center}
	
	\caption{Quantum Circuit}
	\label{fig:quantum_circuit}	
	\end{subfigure}\hfil
	\begin{subfigure}[b]{0.48\linewidth}
		\begin{center}
	\qquad
	\qquad
	\mbox{	
		\Qcircuit @C=0.5em @R=0.75em @!R {
			\lstick{q_0} & \ustick{0} \qw & \targ & \ustick{1} \qw & \ctrl{1} & \ustick{1} \qw & \targ     & \ustick{1} \qw & \ctrl{2} & \ustick{1} \qw & \qw \\
			\lstick{q_1} & \ustick{0} \qw & \qw   & \ustick{0} \qw & \targ    & \ustick{1} \qw & \ctrl{0}  & \ustick{1} \qw & \ctrl{0} & \ustick{1} \qw & \qw \\
			\lstick{q_2} & \ustick{0} \qw & \qw   & \ustick{0} \qw & \qw      & \ustick{0} \qw & \ctrl{-2} & \ustick{0} \qw & \targ    & \ustick{1} \qw & \qw \\
	}}
	
\end{center}

	\caption{Reversible Circuit}
	\label{fig:reversible_circuit}		
	\end{subfigure}

	\caption{Circuit Diagrams}
\end{figure}
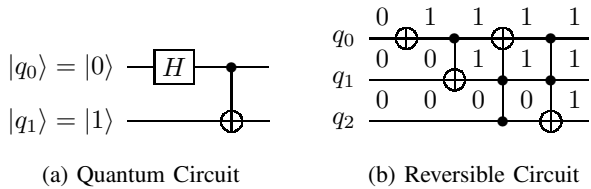

\begin{example}
	The quantum circuit shown in Fig.~\ref{fig:quantum_circuit}	is composed of two qubits and two gates. First, a Hadamard operation is applied to qubit $q_0$, setting $q_0$ into a superposition. Afterwards, a controlled NOT (CNOT) operation is conducted, where $q_0$ serves as control qubit and $q_1$ is the target qubit. Here, the value of $q_1$ is inverted if $q_0$ is in the basis state $\ket{1}$.
\end{example}

Reversible circuits are a subset of quantum circuits that can be modeled in the classical domain. 
Hence, these circuits are used when designing Boolean components for quantum circuits and are usually composed of multiple-controlled Toffoli gates. 
These gates are composed of a (possibly empty) set of control qubits and a so-called target qubit. The value of the target qubit is inverted if, and only if, all control qubits are in basis state $\ket{1}$. Hence, the CNOT gate discussed above is a multiple-controlled Toffoli gate with a single control.

\begin{example}
	Consider the reversible circuit shown in Fig.~\ref{fig:reversible_circuit} that is composed of three qubits and four gates. The target qubit of a Toffoli gate is again denoted by $\oplus$, whereas control qubits are denoted by $\bullet$.	
	Additionally, we have labeled the intermediate values of the qubits throughout the circuit when applying $\ket{q_0q_1q_2} = \ket{000}$ as input. Since a reversible circuit can be modeled in the classical domain (no quantum effects are exploited), we use $0$ and $1$ to indicate the basis states (rather than $\ket{0}$ and $\ket{1}$).
	The first gate has not control qubits and, thus, inverts the value of $q_0$ from 0 to 1. The second gate is controlled by $q_0$. Since the value of $q_0$ is 1, the value of $q_1$ is inverted from 0 to 1. The second gate does not affect the state of the qubits, since the control qubit $q_2$ is set to 0. Eventually, the last gate inverts the state of $q_2$ from 0 to 1.
\end{example}

\subsection{Boolean Components in Quantum Circuits}
\label{sec:embedding}

Typically, quantum circuits contain large Boolean components (also called \emph{oracles}) which can be realized by \emph{reversible circuits}. Decompositions of the gates occurring in reversible circuits into elementary quantum operations (e.g.,~into the well-known \emph{Clifford+T} library~\cite{boykin2000new}) can be determined using approaches such as~\cite{DBLP:journals/tcad/AmyMMR13}.

Since Boolean components occurring in quantum circuits commonly describe very complex functionality, they are usually split into several \mbox{non-reversible} parts (e.g.,~the modular exponentiation in Shor's algorithms can be build up from several adders)---either manually~\cite{DBLP:journals/qic/HanerRS17,DBLP:journals/qic/Beauregard03,haener2018quantum} or by automated synthesis tools (using methods as reviewed, e.g., in~\cite{DBLP:conf/ismvl/DrechslerW11,DBLP:journals/csur/SaeediM13}). But since quantum computations are inherently reversible, it has to be ensured that these sub-components are realized in a reversible fashion, i.e.,~as a function realizing a unique mapping from the inputs to the outputs and vice versa.

\begin{example}
	Consider the truth table of a half adder shown in Table~\ref{tab:tt_ha} and assume that this functionality shall be realized as a sub-function of an oracle. 
	Since the output pattern 01 occurs twice, the function is not reversible---the input cannot be determined uniquely having the output only.
\end{example}

\begin{table}
	\caption{Truth table of the half adder function}
	
	\begin{subtable}[t]{0.48\linewidth}
		\caption{Before embedding}
		\label{tab:tt_ha}
		
	\centering
		\begin{tabular}{cc|cc}
			$x_1$ & $x_2$ & $y_1$ & $y_0$ \\\hline
			0 & 0 & 0 & 0 \\
			0 & 1 & 0 & 1 \\
			1 & 0 & 0 & 1 \\
			1 & 1 & 1 & 0 \\
		\end{tabular}
	\end{subtable}
	\begin{subtable}[t]{0.48\linewidth}
		\caption{After embedding}
		\label{tab:tt_ha_embedded}

	\centering
	\begin{tabular}{ccc|ccc}
		$a$ & $x_1$ & $x_2$ & $g$ & $y_1$ & $y_0$ \\\hline
		0 & \bfseries 0 & \bfseries 0 & 0 & \bfseries 0 & \bfseries 0 \\
		0 & \bfseries 0 & \bfseries 1 & 1 & \bfseries 0 & \bfseries 1 \\
		0 & \bfseries 1 & \bfseries 0 & 0 & \bfseries 0 & \bfseries 1 \\
		0 & \bfseries 1 & \bfseries 1 & 1 & \bfseries 1 & \bfseries 0 \\
		1 & 0 & 0 & 0 & 1 & 0 \\
		1 & 0 & 1 & 0 & 1 & 1 \\
		1 & 1 & 0 & 1 & 0 & 0 \\
		1 & 1 & 1 & 1 & 1 & 1 \\
	\end{tabular}
	
	\end{subtable}
\end{table}

To ensure a unique input-output mapping, the non-reversible function to be realized is \emph{embedded} into a reversible one that typically has a much larger number of variables.\footnote{Note that each variable of the function is realized by means of a qubit in the quantum circuit.}
This embedding process can either be conducted explicitly~\cite{Soeken:2015:ELB:2856147.2786982,ZulehnerW17Emb} (required when using synthesis approaches such as~\cite{SWH+:2012,ZW2017RC}) or implicitly (using synthesis schemes following one-pass synthesis as employed in~\cite{zulehner2017one,ZulehnerW18Exploiting}). However, conducting the embedding often yields circuits where the number of additional variables and, hence, qubits is significant. 
Since qubits are a limited resource (especially in NISQ devices~\cite{preskill2018quantum}) their number shall be kept as small as possible. 
But even following the state of the art reviewed above, still a rather substantial number of qubits results.
In fact, the minimal number of qubits required for embedding thus far is defined as follows:

\begin{definition}
	Consider a Boolean function $f\colon\B^n\rightarrow \B^m$ with output patterns $p_1,p_2,\ldots,p_k \in  \B^m$ ordered by the number of corresponding input patterns (in the following denoted as $\mu(p_i)=|\{x \in \B^n \mid f(x)=p_i\}|$). Since the embedding process has to make all output patterns distinguishable, at least $\lceil\log_2\mu(p_1)\rceil$ additional so-called garbage outputs are required (where $p_1$ is the most frequently occurring output pattern). Moreover, since the number of inputs and outputs has to be equal to realize a reversible function as quantum circuit, a total of $\min(n, m+\lceil\log_2\mu(p_1)\rceil)$ qubits are required to embed a function $f\colon\B^n\rightarrow \B^m$. If this implies to add further inputs, the desired output is obtained when setting all ancillary inputs to a specific value (usually 0).
\end{definition}

\addtocounter{example}{-1}
\begin{example}[continued]
	Since the most frequently occurring output pattern $p_1 = 01$ occurs twice, $\lceil\log_2 2\rceil = 1$ garbage output is required to make this output pattern distinguishable. To align the number of inputs with the number of outputs, one ancillary input is required. 
	Table~\ref{tab:tt_ha_embedded} shows one possible embedding of the half adder function. The desired function can be obtained at the primary outputs by setting the ancillary input $a$ to~0 (highlighted in bold). All garbage variables as well as the primary outputs when $a\neq 0$ can be chosen arbitrarily as long as a reversible function results.
\end{example}

The ancillary qubits of all sub-functions of an oracle have to be de-computed to their initial state to allow for a correct execution within the oracle and  enable a later \mbox{reuse}.

\section{One Ancillary Qubit is Enough}
\label{sec:one_is_enough}

The authors of~\cite{ZulehnerW18Exploiting,zulehner2018pushing} have shown that it is possible to undercut the theoretical lower bound on the number of required qubits (discussed in Section~\ref{sec:embedding}) by using coding techniques, i.e.,~by using a 1-to-1 mapping of the output patterns to others. 
In this section, we first review the main idea of this approach and then formally prove that,
using a variable-length encoding, at most one ancillary qubit is enough to realize any desired non-reversible function---and, by this, any sub-component of an oracle. 
Afterwards, in Section~\ref{sec:results}, it is experimentally confirmed that this indeed allows one to significantly reduce the number of overall required qubits (even below the minimum considered thus far) in many cases.

\subsection{Utilizing Coding}
\label{sec:coding}

As shown in~\cite{ZulehnerW18Exploiting,zulehner2018pushing}, the number of additionally required output patterns can be significantly reduced by exploiting coding techniques. 
The general idea for using coding is motivated by the fact that usually not all output patterns occur equally many times and, thus, do not require the same number of garbage outputs. Hence, a variable-length encoding can be utilized, where frequently occurring output patterns are represented by a short code word (together with a large number of garbage outputs) and rarely occurring output patterns are represented by a longer code word (together with a smaller number of garbage outputs).

\begin{example}
	Consider the Boolean function with $n=3$ inputs and $m=3$ outputs shown in Table~\ref{tab:tt}. Using an embedding scheme as discussed in Section~\ref{sec:embedding} yields a reversible function with five variables (thus, requiring five qubits). However, using the code as shown in Table~\ref{tab:code} allows one to reduce the number of required qubits to three. For example, the most frequently occurring output pattern $p_1 = 110$ (which requires $\lceil \log_2 4 \rceil=2$ garbage outputs) is encoded as $code(p_1) = 0$, while the output pattern $p_3=100$ is encoded by $code(p_3) = 110$. The number of variables/qubits required for each output pattern is then determined by the sum of the code length and the number of required garbage outputs---resulting in the encoded function shown in Table~\ref{tab:coded_function} (dashes indicate garbage variables).
\end{example}

\begin{table}
	\caption{Encoding a non-reversible function}
	
	\centering
	\begin{subtable}[t]{0.30\linewidth}		
		\caption{Orig. function}
		\label{tab:tt}
		\centering
		\scriptsize
		\setlength{\tabcolsep}{1pt}
		\begin{tabular}{ccc|ccc}
			$x_3$ & $x_2$ & $x_1$& $x_3'$ & $x_2'$ & $x_1'$ \\\hline
			0 & 0 & 0 & 1 & 1 & 0\\
			0 & 0 & 1 & 0 & 0 & 0\\
			0 & 1 & 0 & 1 & 1 & 0\\
			0 & 1 & 1 & 1 & 0 & 0\\
			1 & 0 & 0 & 0 & 0 & 0 \\
			1 & 0 & 1 & 1 & 1 & 1 \\
			1 & 1 & 0 & 1 & 1 & 0\\
			1 & 1 & 1 & 1 & 1 & 0\\
		\end{tabular}
	\end{subtable}	
	\begin{subtable}[t]{0.30\linewidth}		
		\caption{Encoding}
\label{tab:code}
\centering
\scriptsize
\setlength{\tabcolsep}{3pt}
\begin{tabular}{c|c|c|c}
	$i$ & $p_i$ & $\mu(p_i)$ & $code(p_i)$ \\\hline
	1 & 110 & 4 & 0 - - \\
	2 & 000 & 2 & 1 0 - \\
	3 & 100 & 1 & 1 1 0 \\
	4 & 111 & 1 & 1 1 1 \\
\end{tabular}		
	\end{subtable}
\begin{subtable}[t]{0.38\linewidth}
	\caption{Encoded function}
\label{tab:coded_function}
\centering
\scriptsize
\setlength{\tabcolsep}{1pt}

\begin{tabular}{ccc|ccc}
	$x_3$ & $x_2$ & $x_1$& $x_3'$ & $x_2'$ & $x_1'$ \\\hline
	0 & 0 & 0 & 0 & - & -\\
	0 & 0 & 1 & 1 & 0 & -\\
	0 & 1 & 0 & 0 & - & -\\
	0 & 1 & 1 & 1 & 1 & 0\\
	1 & 0 & 0 & 1 & 0 & -\\
	1 & 0 & 1 & 1 & 1 & 1\\
	1 & 1 & 0 & 0 & - & -\\
	1 & 1 & 1 & 0 & - & -\\
\end{tabular}	
\end{subtable}

\vspace*{-3mm}
\end{table}

To generate a code as shown above, a \emph{Pseudo-Huffman encoding} is employed. To this end, one starts with terminal nodes---one for each output pattern with $\mu(p_i) > 0$ (no code has to be assigned to output patterns that do not occur)---and attaches a weight representing the number of  required garbage outputs (i.e.,~$\lceil\log_2\mu(p_i)\rceil$). The \emph{Pseudo-Huffman tree} is then generated by repeatedly combining the two nodes $a$ and $b$ with the smallest attached weights $w(a)$ and $w(b)$ to a new node $c$ with attached weight \mbox{$w(c) = \max(w(a), w(b))+1$} until a single node results. The weight of such a node $w(c)$ then gives the number of outputs required to represent all combined output patterns uniquely, i.e.,~one additional variable is required (aside from $\max(w(a), w(b))$) to distinguish between $a$ and $b$. Hence, the weight of the root node determines the number of overall required outputs in the encoded function.
Building the \emph{Pseudo-Huffman tree} inherently gives such a variable-length encoding of the output patterns by, e.g.,~assigning 0 (1) to the left (right) successor of each node. Concatenation of the values attached to the path from the root node to a terminal representing an output pattern $p_i$ determines $code(p_i)$. 

\begin{example}
	Figure~\ref{fig:huffman} shows the Pseudo-Huffman tree for the function shown in Table~\ref{tab:tt}. Since there exist four output patterns with $\mu(p_i) > 0$, we start with four terminal nodes (labeled $v_1$, $v_2$, $v_3$, and $v_4$, respectively) and attach the number of required garbage outputs as weights (drawn as numbers inside the nodes). First, we combine the nodes $v_3$ and $v_4$ to a new node $v_5$ with weight $w(v_5) = \max(0,0)+1=1$. Next, we combine the nodes $v_2$ and $v_5$ to a new node $v_6$ with weight $w(v_6) = \max(1,1)+1=2$. Eventually, the nodes $v_1$ and $v_6$ are combined to a node $v_7$ with weight $w(v_7) = \max(2,2)+1=3$---the single root node of the tree. The code for the individual output patterns is then determined by the path from the root node to the respective terminal. For example, output pattern $p_{2}=000$ is encoded by $code(p_{2}) = 10$ since the path traverses the right edge of node $v_7$ and the left edge of node $v_6$. Overall, the code shown in Table~\ref{tab:code} results. 
\end{example}

\begin{figure}
	\centering

		\begin{tikzpicture}[terminal/.style={draw,rectangle,inner sep=2pt}]
		\matrix[matrix of nodes,ampersand replacement=\&,every node/.style={vertex},column sep={1.5cm,between origins},row sep={1cm,between origins}] (qmdd) {
			\node (n1) {$3$}; \& \& \& \\
			\& \node(n2) {$2$}; \& \& \\
			\& \& \node (n3) {$1$}; \& \\
			\node (n4a) {$2$}; \& \node (n4b) {$1$}; \& \node (n4c) {$0$}; \& \node (n4d) {$0$}; \\
		};
		
		\draw (n1) -- (n4a) node[midway, left] {0};
		\draw (n1) -- (n2) node[midway, right] {1};
		
		\draw (n2) -- (n4b) node[midway, left]{0};
		\draw (n2) -- (n3) node[midway, right]{1};
		
		\draw (n3) -- (n4c) node[midway, left]{0};
		\draw (n3) -- (n4d) node[midway, right]{1};
		
		\node[right=0cm of n4a] {$v_1$};
		\node[right=0cm of n4b] {$v_2$};
		\node[right=0cm of n4c] {$v_3$};
		\node[right=0cm of n4d] {$v_4$};
		\node[right=0cm of n3] {$v_5$};
		\node[right=0cm of n2] {$v_6$};
		\node[right=0cm of n1] {$v_7$};
		
		\node[below=0cm of n4a] {$p_1=110$};
		\node[below=0cm of n4b] {$p_2=000$};
		\node[below=0cm of n4c] {$p_3=100$};
		\node[below=0cm of n4d] {$p_4=111$};
		
		\end{tikzpicture}

	\caption{Pseudo-Huffman tree for the function from Table~\ref{tab:tt}}
	\label{fig:huffman}
	\end{figure}
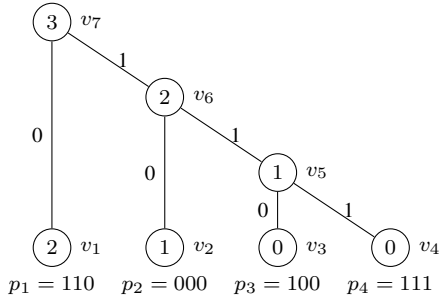

\subsection{Proving an Upper Bound of $n+1$ Qubits}
\label{sec:proof}

In this section, we prove that encoding the output patterns of an $n$-input function as shown above results in a coded function requiring at most $n+1$ variables. Moreover, we show precisely in which cases this additional qubit is required and in which not. To this end, we first formally define the Pseudo-Huffman tree utilized to determine the encoding.

\begin{definition}[Pseudo-Huffman Tree]
Let $G = (V,E)$ be a connected, arborescence (a directed rooted tree) composed of a set of nodes $V = \{v_1, v_2, \ldots, v_{|V|}\}$ and a set of edges \mbox{$E \subset V \times V$}, and let $w\colon V\to \mathbb{N}_0$ be a labeling of the graph nodes in terms of non-negative weights. Moreover, let $T=\{t\in V\mid\forall v \in V\colon (t,v) \notin E\} \subseteq V$ denote the set of all terminal nodes.
Then, $PH=(G, w)$ is called a \emph{Pseudo-Huffman tree}, if, and only if, \begin{enumerate}
\item each internal node $v \in V\setminus T$ has exactly two children $a,b\in V$ and $w(v) =\max(w(a),w(b))+1$ and 
\item for any two different internal nodes $v_1,v_2 \in V\setminus T$ with children $a_1,b_1$ and $a_2,b_2$, respectively, it holds that $w(a_1) \le w(b_1)$
 implies
\[ \big(w(a_2), w(b_2) \leq w(a_1)\big) \lor \big(w(a_2), w(b_2) \geq w(b_1)\big).\quad\quad\quad \]
\end{enumerate}
\end{definition}

In other words, the tree can be formed from the terminal nodes by successively combining nodes with the lowest available weights as described in Section~\ref{sec:coding}. 

The following theorem yields a condition on the terminal nodes of a Pseudo-Huffman tree that is sufficient to restrict the weight of the tree's root node.

\begin{theorem}
\label{theorem1}
Let $PH=((V,E),w)$ be a Pseudo-Huffman tree. If there exists an assignment $s_v$ for each terminal node $v\in T=\{t\in V\mid\forall v \in V\colon (t,v) \notin E\}$ such that \mbox{$2^{w(v)} \geq s_v > 2^{w(v)-1}$} (where $w(v)$ denotes the weight of node~$v$) and $\sum_{v\in T} s_v = 2^n$, then the weight $w(v_r)$ of the root node $v_r$ of the tree is either $n$ or $n+1$.
\end{theorem}

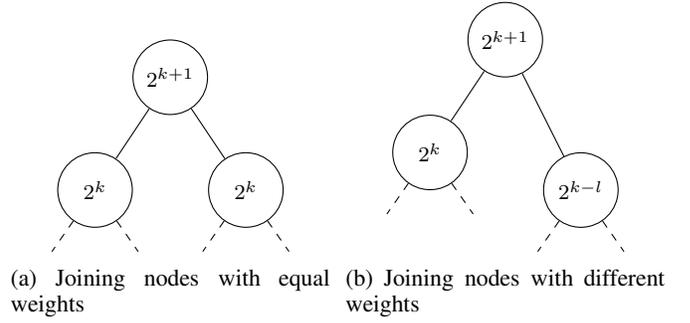
\begin{figure}
	\centering
	\begin{subfigure}[b]{0.48\linewidth}
		\centering
		\begin{tikzpicture}
		
		\node [circle, draw] (a) {$\phantom{2^{k+1}}$};
		\node [circle, draw, right=of a] (b) {$\phantom{2^{k+1}}$};
		\node [circle, draw] (c) at ([xshift=1cm,yshift=1.5cm] a) {$2^{k+1}$};
		\node at (a) {$2^k$};
		\node at (b) {$2^k$};
		\draw [dashed] (a) -- ++(235:1cm);
		\draw [dashed] (a) -- ++(305:1cm);
		\draw [dashed] (b) -- ++(235:1cm);
		\draw [dashed] (b) -- ++(305:1cm);
		\draw (a)--(c)--(b);
		\end{tikzpicture}

		\caption{Joining nodes with equal weights}
		\label{fig:nodeCombinationSameWeight}
	\end{subfigure}	\hfil
	\begin{subfigure}[b]{0.48\linewidth}
		\centering
		\begin{tikzpicture}
		\node [circle, draw] (a) {$\phantom{2^{k+1}}$};
		\node [circle, draw, right=of a, yshift=-5mm] (b) {$\phantom{2^{k+1}}$};
		\node [circle, draw] (c) at ([xshift=1cm,yshift=1.5cm] a) {$2^{k+1}$};
		\node at (a) {$2^k$};
		\node at (b) {$2^{k-l}$};
		\draw [dashed] (a) -- ++(235:1cm);
		\draw [dashed] (a) -- ++(305:1cm);
		\draw [dashed] (b) -- ++(235:1cm);
		\draw [dashed] (b) -- ++(305:1cm);
		\draw (a)--(c)--(b);
		\end{tikzpicture}

	\caption{Joining nodes with different weights}
	\label{fig:nodeCombinationDifferentWeight}
	\end{subfigure}
	\caption{Joining nodes in the construction of the PH-tree}
\end{figure}

\begin{proof}
	Replace all weights using the rule $w\mapsto 2^w$. Then the rule for computing the weight of a new node changes from $\max(w(a),w(b))+1$ to $2\cdot\max(w(a),w(b))$. Accordingly, all weights in the tree will be a power of 2.
	
	We perform the proof by arguing about the weights of the nodes when constructing a Pseudo-Huffman tree. To this end, consider the set of all nodes $V^i_r$ of the tree-under-construction that are the root nodes of the already connected components after step $i$ of the algorithm. Let  $w^i_{total}=\sum_{v\in V_r^i} w(v)$ denote the sum of the weights over all these nodes.
	
	At each step $i$ of the algorithm, two nodes $a,b \in V^i_r$ with minimal weight are chosen and joined to a new node $c$ such that $V^{i+1}_r = \{c \}\cup V^i_r \setminus\{a, b\}$.
	There are two cases: 
	\begin{enumerate}
		\item both nodes $a$ and $b$ have the same weight $2^k$. Then, they are replaced by a node with weight $2^{k+1}$ such that 
		$w_{total}^{i+1} = w_{total}^i$ (see Fig.~\ref{fig:nodeCombinationSameWeight}), i.e.,~the sum of the weights over the root nodes remains constant.
		\item one node---assume without loss of generality $a$---has weight $w(a)=2^k$ and the other node ($b$) has weight $w(b)=2^{k-l}$ for some $k\ge l>0$. Then, they are replaced by a node $c$ with weight $w(c)=2^{k+1}$ (see Fig.~ \ref{fig:nodeCombinationDifferentWeight}).
		
		Since we always take the nodes with minimal weight, there might not be any other node $d \in V^i_r$ with \mbox{$w(d)<2^k$} as this node would have a higher priority to be joined with $b$. Thus, all nodes in $V^i_r$ aside from $b$ have a weight that---by construction---is a power of 2 that is greater than or equal to $2^k$. Consequently, after joining $a$ and $b$, $w^i_{total}$ is increased to a number $w^{i+1}_{total}$ that is divisible by~$2^k$. More precisely, it is increased by
		
		\begin{align*}
		w(c)-w(a)-w(b) &= 2^{k+1}-2^k-2^{k-l}\\
		&= 2^k-2^{k-l}\\
		&< 2^k,
		\end{align*}
		such that  $w^{i+1}_{total}$ is the \emph{smallest} number that is greater than $w^i_{total}$ and divisible by $2^{k}$.
		
		Clearly, this case happens at most once for each $k>0$, since afterwards there is no more node in $V^{i+1}_r$ with a weight less than $2^k$ and all nodes that will be added to $V^{j}_r$ (for $j>i$) have higher weights.
	\end{enumerate}
	
	\vspace{200cm}
	
	By the assumption of Theorem~\ref{theorem1}, we initially have \mbox{$2^n = \sum_{v\in T} s_v \leq  w^0_{total}$} and $w^0_{total} < 2^{n+1}$.
	Thus, we will at some point denoted $final$ reach the case that all nodes in $V^{final}_r$ have a weight greater than or equal to $2^n$ such that $w^{final}_{total}$ is divisible by $2^n$.  
	Since $2^{n+1}$ is divisible by all potencies $2^k$ for $k=0,\ldots,n$, $w^{final}_{total}$ will never exceed $2^{n+1}$, as we are always increasing $w^i_{total}$ to the smallest larger number divisible by $2^{k}$ for a $k\in\{1,\ldots,n\}$.
	Consequently, we have at least one and at most two nodes in $V^{final}_r$ with a weight of $2^n$. Thus, the root node of the resulting tree either is the single node with weight $2^n$ or the single node with weight $2^{n+1}$ constructed from the two nodes with weight $2^n$.
	Hence, the root node of the original Pseudo-Huffman tree has weight $n$ or $n+1$ as desired.
\end{proof}

Now let us interpret this result in the setting of coded Boolean functions.
Consider a Boolean function $f: \B^n\rightarrow \B^m$ to be encoded. 
We can construct a Pseudo-Huffman tree with $|T| = |\{p_i \in \B^m\mid \mu(p_i)>0\}|$ terminal nodes (which is always possible), where each terminal node $v\in T$ uniquely corresponds to one output pattern $p_i$ and has assigned $s_v = \mu(p_i)$ (thus, having a weight $w(v) = \lceil \log_2 \mu(p_i)\rceil$).
As this assignment clearly satisfies the conditions of Theorem~\ref{theorem1}, the height of this tree is either $n$ or $n+1$. Hence, there exists a coding (which is inherently given by the constructed tree) that requires at most one additional qubit when realizing $f$ in quantum logic.

Moreover, we can precisely determine in which cases this additional qubit is required. In fact, the additional qubit is required whenever there exists an output pattern $p_i$ where $\mu(p_i) > 0$ is not a power of two.

\begin{corollary} 
The root node of a Pseudo-Huffman tree satisfying the same assumptions as in Theorem~\ref{theorem1} has weight $n$ if, and only if, 
$\sum_{v\in T} 2^{w(v)} = 2^n$.
\end{corollary}

\begin{proof}
Given that $\sum_{v\in T} 2^{w(v)} = 2^n$ we may apply Theorem~\ref{theorem1} by using the assignment $s_v=2^{w(v)}$ for all $v \in T$.
Following the argumentation in the proof of Theorem 1, the root node of the Pseudo-Huffman tree has weight $n$ if $w^{i}_{total}$ does not exceed $2^n$ at any time, i.e.,~if the second case (which increases $w^i_{total}$) does not occur at all. This is clearly the case if $w^0_{total}=2^n$ in the beginning.
Conversely, if \mbox{$\sum_{v\in T} 2^{w(v)}\neq 2^n$}, we have $w^0_{total}>2^n$ in the beginning, such that $w^{final}_{total}=2^{n+1}$ in the end.
\end{proof}

\section{Comparison to Embeddings without Coding}
\label{sec:results}

In this section, we compare the idea of coded embeddings to previous approaches and discuss their effect on the design of quantum oracles.

\subsection{Evaluation}

We compare the idea of coded embeddings to approaches that do not consider coding when realizing a Boolean function $f\colon\B^n\rightarrow \B^m$ in quantum logic. More precisely, we compare to exact methods utilizing $\max(n, m+\lceil \log_2 \mu(p_1)\rceil)$ qubits~\cite{Soeken:2015:ELB:2856147.2786982,ZulehnerW17Emb} as well as to heuristic ones that always utilize a Bennett embedding with $n+m$ qubits~\cite{Bennett:73,WKD:2011} (e.g.,~generated when using an ESoP based synthesis approach~\cite{FTR:2007}). To this end, we have implemented the proposed idea in C++ and utilized the QMDD package~\cite{DBLP:journals/tcad/NiemannWMTD16} as well as the BDD package CUDD~\cite{somenzi2015cudd} to gain a compact representation of the considered functions---allowing us to determine the number of required qubits in negligible runtime. As benchmarks we use the functions from RevLib~\cite{WGT+:2008}, as well as from the ISCAS~\cite{iscas} and IWLS~\cite{McE:93} benchmark suites.\footnote{Note that we only consider non-reversible functions from these benchmarks suits since reversible ones do not require embedding.}

Table~\ref{tab:results} summarizes the obtained results. The first three columns of the benchmark as well as the number of inputs $n$ and the number of outputs $m$. In the next three columns we list the number of required qubits when using Bennett embedding (i.e.,~$m+n$), when using a minimal encoding without considering coding (i.e.,~$\max(n, m+\lceil\log_2\mu(p_1)\rceil)$), and when using coded embeddings as described in this work (i.e.,~$n$ or $n+1$), respectively.

\begin{table}
	\caption{Number of required qubits}
	\label{tab:results}
	
	\setlength{\tabcolsep}{5pt}
	\centering
	\begin{tabular}{l|r|r||r|r||r}
		\multicolumn{3}{c||}{Benchmark} & \multicolumn{3}{c}{Embedding} \\
		name & $n$ & $m$ & Bennett~\cite{Bennett:73} & Min.~\cite{Soeken:2015:ELB:2856147.2786982,ZulehnerW17Emb} & Encoded \\\hline
		\begin{filecontents*}{results.csv}
file,inputs,outputs,bennett,min,coded
f51m\_159,14,8,22,19,15
tial\_214,14,8,22,19,15
cu\_141,14,11,25,25,15
misex3\_180,14,14,28,28,15
misex3c\_181,14,14,28,28,15
table3\_209,14,14,28,28,15
s1488\_split,14,25,39,38,15
s1494\_split,14,25,39,38,15
b12,15,9,24,22,16
in0\_162,15,11,26,25,16
parity\_188,16,1,17,16,16
ryy6\_198,16,1,17,17,17
t481\_208,16,1,17,17,17
cmb\_134,16,4,20,20,17
pcler8\_190,16,5,21,21,17
cm163a\_133,16,13,29,25,17
pdc\_191,16,40,56,55,17
spla\_202,16,46,62,61,17
table5,17,15,32,32,18
s298\_split,17,20,37,29,18
s208.1\_split,18,9,27,19,19
cm151a\_129,19,9,28,27,20
cm150a\_128,21,1,22,22,22
mux\_185,21,1,22,22,22
duke2,22,29,51,50,23
cordic\_138,23,2,25,25,24
cps\_140,24,109,133,132,25
vg2,25,8,33,32,26
misex2,25,18,43,42,26
frg1\_160,28,3,31,30,29
apex2\_101,39,3,42,42,40
seq\_201,41,35,76,75,42
apex1,45,45,90,89,46
apex3,54,50,104,103,55
e64\_149,65,65,130,129,65
\end{filecontents*}
\csvreader[
		late after line=\\,
		late after last line=\\,
		]{results.csv}
		{1=\Name,2=\In, 3=\Out, 4=\Bennett, 5=\Min, 6=\Coded}
		{\Name & \In & \Out & \Bennett & \Min & \Coded}
		
	\end{tabular}
	\vspace{-4mm}
\end{table}

As can be seen in Table~\ref{tab:results}, the number or required qubits can significantly be reduced when considering coded embeddings---especially in cases where $m>n$. Consider for example benchmarks \emph{cps\_140} and \emph{e64\_149}, where the number of required qubits can be reduced by 107 and 65, respectively, using coding techniques. Overall, a possible reduction of 36.4\% can be observed on average.

\vspace{200cm}

\subsection{Discussion}
Concerning the design of quantum oracles, coded embeddings as proposed above can be exploited in two different ways:
\begin{itemize}
\item On the one hand, one can apply the coding technique \emph{locally} on each and every sub-component and use decoders (after each sub-component) to translate the encoded results to the original ones which are then used as inputs of the subsequent components. This essentially reduces the complexity of synthesis for the individual sub-components (since a smaller number of qubits needs to be considered). While this offers a significant improvement of synthesis run-time (as also observed in~\cite{ZulehnerW18Exploiting}), the total number of additional qubits does not change (due to the decoders).
\item On the other hand, one can apply the coding technique \emph{globally} such that the encoded outputs of one sub-component are directly used as input for subsequent components and a single decoder at the end translates the final results to the desired ones. This approach significantly reduces the number of extra qubits required during the computation of the oracle's sub-components such that the total number of extra qubits is likely to stay close to the theoretical minimum given by the oracle's overall functionality (which is zero). On the downside, a re-design of the sub-components might be required in order to work with encoded values. 
\end{itemize}

\section{Conclusions}
\label{sec:conclusions}

In this work, we have proven that one additional qubit is enough to determine a coded embedding of any non-reversible function for quantum circuits. By this, one can significantly reduce the overall number of qubits required for realizing Boolean oracles, since their functionality is usually split into several non-reversible parts.
Our experimental evaluation shows that the number of required qubits can indeed be reduced by 36.4\% on average, when comparing to embeddings that do not utilize encoding and that have been considered as the minimum thus far. Possible applications on the design of oracles for quantum circuits are discussed.

\section*{Acknowledgements}
This work has partially been supported by the European Union through the COST Action IC1405.

\bibliographystyle{unsrt}
{
	\bibliography{literature}

\begin{thebibliography}{10}

\bibitem{NC:2000}
M.~Nielsen and I.~Chuang.
\newblock {\em Quantum Computation and Quantum Information}.
\newblock Cambridge Univ. Press, 2000.

\bibitem{courtland2017google}
Rachel Courtland.
\newblock Google aims for quantum computing supremacy.
\newblock {\em IEEE Spectrum}, 54(6):9--10, 2017.

\bibitem{gomes2018quantumcomputing}
Lee Gomes.
\newblock Quantum computing: Both here and not here.
\newblock {\em IEEE Spectrum April 2018}, 2018.

\bibitem{hsu2018quantumcomputing}
Jeremy Hsu.
\newblock {CES} 2018: Intel's 49-qubit chip shoots for quantum supremacy.
\newblock {\em IEEE Spectrum Tech Talk}, 2018.

\bibitem{preskill2018quantum}
John Preskill.
\newblock Quantum computing in the {NISQ} era and beyond.
\newblock {\em Quantum}, 2:79, 2018.

\bibitem{Sho:94}
P.~W. Shor.
\newblock Algorithms for quantum computation: discrete logarithms and
  factoring.
\newblock {\em Foundations of Computer Science}, pages 124--134, 1994.

\bibitem{Gro:96}
Lov~K. Grover.
\newblock A fast quantum mechanical algorithm for database search.
\newblock In {\em Theory of computing}, pages 212--219, 1996.

\bibitem{DBLP:conf/ismvl/DrechslerW11}
Rolf Drechsler and Robert Wille.
\newblock From truth tables to programming languages: Progress in the design of
  reversible circuits.
\newblock In {\em Int'l Symp. on {M}ulti-{V}alued {L}ogic}, pages 78--85, 2011.

\bibitem{DBLP:journals/csur/SaeediM13}
Mehdi Saeedi and Igor~L. Markov.
\newblock Synthesis and optimization of reversible circuits - a survey.
\newblock {\em {ACM} Comput. Surv.}, 45(2):21, 2013.

\bibitem{DBLP:journals/qic/HanerRS17}
Thomas H{\"{a}}ner, Martin Roetteler, and Krysta~M. Svore.
\newblock Factoring using 2n+2 qubits with {T}offoli based modular
  multiplication.
\newblock {\em Quantum Information {\&} Computation}, 17(7{\&}8):673--684,
  2017.

\bibitem{DBLP:journals/qic/Beauregard03}
St{\'{e}}phane Beauregard.
\newblock Circuit for {S}hor's algorithm using 2n+3 qubits.
\newblock {\em Quantum Information {\&} Computation}, 3(2):175--185, 2003.

\bibitem{haener2018quantum}
Thomas Haener, Mathias Soeken, Martin Roetteler, and Krysta~M Svore.
\newblock Quantum circuits for floating-point arithmetic.
\newblock In {\em International Conference on Reversible Computation}, pages
  162--174. Springer, 2018.

\bibitem{Soeken:2015:ELB:2856147.2786982}
Mathias Soeken, Robert Wille, Oliver Keszocze, D.~Michael Miller, and Rolf
  Drechsler.
\newblock Embedding of large {B}oolean functions for reversible logic.
\newblock {\em J. Emerg. Technol. Comput. Syst.}, 12(4):41:1--41:26, December
  2015.

\bibitem{ZulehnerW17Emb}
Alwin Zulehner and Robert Wille.
\newblock Make it reversible: Efficient embedding of non-reversible functions.
\newblock pages 458--463, 2017.

\bibitem{zulehner2017one}
Alwin Zulehner and Robert Wille.
\newblock One-pass design of reversible circuits: Combining embedding and
  synthesis for reversible logic.
\newblock 37(5):996--1008, 2018.

\bibitem{ZulehnerW18Exploiting}
Alwin Zulehner and Robert Wille.
\newblock Exploiting coding techniques for logic synthesis of reversible
  circuits.
\newblock In {\em Asia and South Pacific Design Automation Conf.}, pages
  670--675, 2018.

\bibitem{zulehner2018pushing}
Alwin Zulehner and Robert Wille.
\newblock Pushing the number of qubits below the “minimum”: Realizing
  compact boolean components for quantum logic.
\newblock In {\em 2018 Design, Automation \& Test in Europe Conference \&
  Exhibition (DATE)}, pages 1179--1182. IEEE, 2018.

\bibitem{zulehner2018generalizing}
Alwin Zulehner, P.~Mercy~Nesa Rani, Kamalika Datta, Indranil Sengupta, and
  Robert Wille.
\newblock Generalizing the concept of scalable reversible circuit synthesis for
  multiple-valued logic.
\newblock In {\em Int'l Symp. on {M}ulti-{V}alued {L}ogic}, pages 115--120,
  2018.

\bibitem{boykin2000new}
P~Oscar Boykin, Tal Mor, Matthew Pulver, Vwani Roychowdhury, and Farrokh Vatan.
\newblock A new universal and fault-tolerant quantum basis.
\newblock {\em Information Processing Letters}, 75(3):101--107, 2000.

\bibitem{DBLP:journals/tcad/AmyMMR13}
Matthew Amy, Dmitri Maslov, Michele Mosca, and Martin Roetteler.
\newblock A meet-in-the-middle algorithm for fast synthesis of depth-optimal
  quantum circuits.
\newblock {\em {IEEE} Trans. on {CAD} of Integrated Circuits and Systems},
  32(6):818--830, 2013.

\bibitem{SWH+:2012}
Mathias Soeken, Robert Wille, Christoph Hilken, Nils Przigoda, and Rolf
  Drechsler.
\newblock Synthesis of reversible circuits with minimal lines for large
  functions.
\newblock In {\em Asia and South Pacific Design Automation Conf.}, pages
  85--92, 2012.

\bibitem{ZW2017RC}
Alwin Zulehner and Robert Wille.
\newblock Improving synthesis of reversible circuits: Exploiting redundancies
  in paths and nodes of {QMDD}s.
\newblock In {\em Int'l Conf. of Reversible Computation}, pages 232--247, 2017.

\bibitem{Bennett:73}
C.H. Bennett.
\newblock Logical reversibility of computation.
\newblock {\em IBM Journal of Research and Development}, 17(6):525--532, Nov
  1973.

\bibitem{WKD:2011}
R.~Wille, O.~Kesz\"ocze, and R.~Drechsler.
\newblock Determining the minimal number of lines for large reversible
  circuits.
\newblock In {\em Design, Automation and Test in Europe}, 2011.

\bibitem{FTR:2007}
K.~Fazel, M.A. Thornton, and J.E. Rice.
\newblock {ESOP}-based {T}offoli gate cascade generation.
\newblock In {\em Communications, Computers and Signal Processing, 2007. PacRim
  2007. IEEE Pacific Rim Conference on}, pages 206 --209, 2007.

\bibitem{DBLP:journals/tcad/NiemannWMTD16}
Philipp Niemann, Robert Wille, D.~Michael Miller, Mitchell~A. Thornton, and
  Rolf Drechsler.
\newblock {QMDDs}: Efficient quantum function representation and manipulation.
\newblock {\em {IEEE} Trans. on {CAD} of Integrated Circuits and Systems},
  35(1):86--99, 2016.

\bibitem{somenzi2015cudd}
Fabio Somenzi.
\newblock {CUDD}: {CU} decision diagram package release 3.0. 0.
\newblock 2015.

\bibitem{WGT+:2008}
R.~Wille, D.~Gro{\ss}e, L.~Teuber, G.~W. Dueck, and R.~Drechsler.
\newblock {RevLib:} an online resource for reversible functions and reversible
  circuits.
\newblock In {\em Int'l Symp. on {M}ulti-{V}alued {L}ogic}, pages 220--225,
  2008.
\newblock {RevLib} is available at http://www.revlib.org.

\bibitem{iscas}
F.~Brglez and H.~Fujiwara.
\newblock {A Neutral Netlist of 10 Combinational Benchmark Circuits and a
  Target Translator in Fortran}.
\newblock In {\em Int'l Symp. Circuits and Systems (ISCAS 85)}, pages 677--692.
  IEEE Press, Piscataway, N.J., 1985.

\bibitem{McE:93}
K.~McElvain.
\newblock {IWLS'93} benchmark set: Version 4.0.
\newblock In {\em Int'l Workshop on Logic Synth.}, 1993.

\end{thebibliography}
}

\end{document}